%% file: main.tex
\documentclass{sig-alternate-techrep}
\usepackage{times}
\usepackage{latexsym,amssymb,amsmath}
\usepackage{subfigure,epic,eepic}
\usepackage{graphicx}
\usepackage{hyperref}
\usepackage{verbatim}
\usepackage{algorithm}
\usepackage{algorithmic}
\usepackage{caption}
\usepackage{luca}
\input{macros}
\begin{document}

\sloppy

\title{CrowdGrader: \\ Crowdsourcing the Evaluation of Homework
Assignments\titlenote{This work was supported in part by the Google Research
Award ``Crowdsourced Ranking''. The authors are listed in alphabetical order.}}
\subtitle{Technical Report UCSC-SOE-13-11 \\ August 2013}

\numberofauthors{2}

\author{
 \alignauthor 
  Luca de Alfaro \\
  \affaddr{Computer Science Dept.}\\
  \affaddr{University of California}\\
  \affaddr{Santa Cruz, CA 95064, USA}\\
  \email{luca@ucsc.edu}
 \alignauthor
  Michael Shavlovsky\\
  \affaddr{Computer Science Dept.}\\
  \affaddr{University of California}\\
  \affaddr{Santa Cruz, CA 95064, USA}\\
  \email{mshavlov@ucsc.edu}
}
\maketitle

\pagenumbering{arabic}
\pagestyle{plain}

\bibliographystyle{abbrv}

\begin{abstract}
\input{abstract}

\end{abstract}

\input{introduction}
\input{previous_work}
\input{crowdgrader}
\input{review}
\input{crowdsourcing}

\input{results}
\input{incentives}

\input{conclusions}

\section*{Acknowledgements}

We thank Ira Pohl at UC Santa Cruz for being an early adopter of CrowdGrader,
and for providing insight and encouragement for this work. 
We thank Marco Faella at the University of Naples for agreeing to use
Crowdgrader in his class when the tool was still in an early, very much
experimental, version.

\bibliography{references}
\end{document}

%% file: macros.tex
\def\algoavg{{\mbox{{\tt avg}}}}
\def\vanc{{\mbox{{\tt vancouver}}}}

\def\V{{\mbox{{\rm var}}}}
\def\E{{\mbox{{\rm E}}}}

%% file: abstract.tex
Crowdsourcing offers a practical method for ranking and scoring large amounts
of items.
To investigate the algorithms and incentives that can be used in crowdsourcing
quality evaluations, we built CrowdGrader, a tool that lets students submit and
collaboratively grade solutions to homework assignments.
We present the algorithms and techniques used in CrowdGrader, and we describe
our results and experience in using the tool for several
computer-science assignments.

CrowdGrader combines the student-provided grades into a consensus grade for
each submission using a novel crowdsourcing algorithm that relies on a reputation system. 
The algorithm iterativerly refines inter-dependent estimates of the consensus
grades, and of the grading accuracy of each student.
On synthetic data, the algorithm performs better than alternatives not based on
reputation.
On our preliminary experimental data, the performance seems dependent on the
nature of review errors, with errors that can be ascribed to the reviewer being
more tractable than those arising from random external events.
To provide an incentive for reviewers, the grade each student receives in an
assignment is a combination of the consensus grade received by their
submissions, and of a reviewing grade capturing their reviewing effort
and accuracy. 
This incentive worked well in practice.

%% file: introduction.tex
\section{Introduction}

Ranking items according to their quality is a universal problem, occurring
when hiring, admitting students, accepting conference papers, presenting
search results, selecting winners in contests, and more. 
Often, the quality of items is best judged by human evaluators.  
As relying on a single evaluator is often impractical --- and can be perceived
as unfair --- an overall evaluation can be obtained via crowdsourcing: several
evaluators compare or grade a subset of the items, and their feedback is then
combined in an overall ranking or scoring of the items.

To study the algorithms and incentives that can be used in crowdsourcing
quality evaluations, we built CrowdGrader, a tool for the crowdsourced
evaluation of homework assignments.
CrowdGrader lets students submit, and collaboratively grade, solutions to
homework assignments; their grade for each assignment depends both on the
quality of their submitted solution, and on the quality of their work as
graders.
CrowdGrader is available at \url{http://www.crowdgrader.org/}.

We chose to focus on homework grading for several reasons.
First, this is a problem that we know well,
and we were confident that the tool would be used by us and by some of
our colleagues, providing valuable experimental data.
Furthermore, solutions submitted to a homework assignment share the same topic:
we do not need to address the problem of matching the topic of each submission
to the domain of expertise of each reviewer, as it is necessary for conference
submissions. 
The students submitting the homework solutions would provide a ready pool of
graders. 
Last, but not least, we hoped the tool would provide educational benefits to
the students. 
We hoped students would benefit from being able to examine the solutions
submitted by other students: accomplished students would be able to look at
alternative ways of solving the same problem, and students who encountered
difficulties would be able to study several working solutions to the problem
while grading.
We also hoped that students would benefit from their peer's feedback.

The first question we studied with the help of CrowdGrader was whether it is
best to ask students to compare and rank submissions, or to assign them
numerical grades. 
Ranking alone did not work to our satisfaction. 
Students expresseded some uneasiness in ranking their peers, especially as they
perceived ranking as a blunt tool, unable to capture the difference between a
pair of roughly equivalent submissions, and a pair of submissions, one of
which was very good, and the other non-functional.
As a consequence, ranking was frequently skipped.
We settled on asking students to both assign grades to each submission, and
rank them in quality order.  
Of course, the ranking can be derived from the grades, but we believe that the
exercise of ranking added precision to the grades. 

The second question concerned the algorithms that can be used to
merge the grades provided by each evaluator, into overall {\em consensus\/}
grades for each assignment.
We developed a novel crowdsourcing algorithm, which we nicknamed \vanc, that
combines the grades provided by the students with the help of a reputation
system that captures the student's grading accuracy. 
The algorithm proceeds via iterations, following a structure inspired by
\cite{Karger2011}, and inspired also by expectation maximization techniques
\cite{dempster1977maximum,dawid1979maximum,raykar_2010}.
In each iteration, the algorithm computes a consensus estimate of the
grade of each submission, weighing the student input according to the accuracy of each
student; the consensus estimates are then used to update the estimated accuracy
of the students.
On synthetic data, \vanc\ performs well, far outperforming algorithms such as
the average or median.
On our real-world data, the results are mixed.  
Our impression is that \vanc\ outperforms simpler algorithms when the grading
errors of the students are not random. 
However, in one of the classes where CrowdGrader was used, the grading errors
were random in nature, due to mis-matches between the code compilation
environments of the students submitting and evaluating homework solutions, and
in that case, \vanc\ performed slightly worse than simple average.

The last question we studied concerned the incentives necessary to obtain
quality evaluations from the students.
Our approach was simple: we made the grade each student received depend on both
the quality of the solution they submitted, and on the quality of their review
and grading work.
This worked well in practice, and we will describe the methods we used
for assigning grading credit.

The remainder of the paper presents in detail the techniques and algorithms
used, and the experimental results obtained.

%% file: previous_work.tex
\section{Previous work}

The work most closely related in goals to ours is the proposal to crowdsource
the review of proposals for use of telescope time by \cite{Merrifield09}, as
well as the recent NSF pilot project for reviewing funding proposals
\cite{NSFDear2013}. 
As in those approaches, we also distribute the task of reviewing the
submissions to the same set of people who submitted the items to be reviewed. 
Both for proposals submitted to a specific panel, and for solutions submitted
to the same homework assignment, the submissions are on sufficiently related
topics that the problem of matching submission topic with reviewer expertise
can be disregarded. 
For proposals, of course, care must be taken to avoid conflicts of interest;
our situation for homeworks is relatively simpler. 
Where the problems differ is that proposal reviewing is essentially a top-$k$
problem: the best $k$ proposals must be selected for funding. 
Homework grading, on the other hand, is an evaluation problem: each item needs
to be graded on a scale. 
In top-$k$ problems, the most important consideration is precision at the top;
mis-ranking items that are far from the top-$k$ carries no real consequence.
In our evaluation problem, each evaluation carries approximately the same
importance, and we do not need to precisely rank students whose
submissions have approximately the same quality. 
While there are techniques that can be applied to both problems, this
difference in goals justifies the reliance of \cite{Merrifield09,NSFDear2013}
on comparisons, and ours on grades.
Comparisons can allow the precise determination of the top-$k$ items in a
ranking \cite{das2010ranking}; we chose instead to develop reputation-based
algorithms for merging grades.

The works of \cite{Merrifield09} and \cite{NSFDear2013} discuss incentive
mechanisms for reviewers, consisting in awarding a better placement in the
final ranking to proposals whose authors did a better job of reviewing. 
We follow the same approach, but we have the additional constraint that
students must find the reviewing work appropriately rewarded with respect to
the time it takes. 
Students are most often under time pressure, and they often consider the
question of whether one hour is better spent reviewing for one class, or
working on the homework assignment of another.
A reward such as the one of \cite{NSFDear2013}, where a couple of places in the
ranking are awarded based on reviewing work, would not have sufficed,
especially in the context of an evaluation rather than top-$k$ setting.
Rather, we let instructors chose a reward magnitude that is commensurate with
the time required by reviewing. 
Furthermore, unlike \cite{Merrifield09,NSFDear2013}, we face the additional
constraint that students must regard the reward as fair and non-punitive; as we
will see in Section~\ref{sec-incentives}, this affected our choice of reward
metrics.

The effect of review incentive on the quality of the ranking is examined in
depth in \cite{Naghizadeh2013}. 
The main problem, also raised in \cite{Merrifield09}, is that the incentive
mechanism makes the grading a ``Keynesian beauty contest", where reviewers
are rewarded for thinking like other reviewers; in turn, this may encourage a
``race to mediocrity'', in which non-controversial, blander propsals may fare
better than more audacious and original ones. 
We agree with the authors of \cite{Naghizadeh2013} that this may be a true
problem for proposal review. 
However, we believe that in the context of homework assignments, the problem
may be minor or non-existent. 
Our incentive function, described in Section~\ref{sec-incentives}, gives a
farily generous reward that would not overly decrease if a students mis-ranks
one of the assignments; this gives more leeway to students presented with a
homework submission that does not follow the beaten path. 
We also believe that the less competitive evaluation setting, as compared to a
top-$k$ setting, may lessen the problem.
Finally, in our somewhat limited real-world experience, students generally
were more ready to reward originality than teaching assistants.
The main goal of a teaching assistant is often to avoid controversy, in order
to avoid confrontations with students.
Thus, teaching assistants generally felt a stronger obligation to follow a
rigid grading scheme, for the sake of consistency, and subtract a fixed number
of points for each type of error encountered.
Students felt less constrained by the need for full consistency, as the
authors of the submissions they graded could not easily identify or compare
the grades they received from the same grader.

The reputation-based crowdsourcing algorithm we use to aggregate grades is
inspired by the algorithm of~\cite{Karger2011} for the aggregation of boolean
input. 
Unlike that work, however, we do not have a proof of convergence for our
crowdsourcing algorithm, nor a full theoretical characterization of how the
precision depends asymptotically on the number of reviews. 
The algorithm is also inspired, and related, to the technique of expectation
maximization
\cite{dempster1977maximum,dawid1979maximum,smyth_94,jin2002learning,whitehill2009whose,raykar_2010,welinder_10}.
The approach is also related to belief propagation methods
\cite{Pearl1988,Yedidia2003}. 
A related, but coarser, method was used by one of the authors to aggregate
information provided by editors of Google Maps via the Crowdsensus system
\cite{dealfaro2011reputation}.

Rank aggregation methods have a very long history.
The problem originally arose in the context of elections. 
In a classical contribution \cite{Borda}, de~Borda proposed
that each voter assigns each of $n$ candidates a score $1, 2, \ldots, n$,
according to the preference; the candidates were then ranked
according to the total score they received from all voters. 
Again in the context of elections, Arrow proved a famous theorem, stating that
any rank aggregation that satisfies transitivity, unanimity, and independence
of irrelevant alternatives is a dictatorship, where there is a single fixed
voter (the dictator) who determines the outcome
\cite{arrow,geanakoplos2005three}.
An overview of rank aggregation methods used in democracies around the world
can be found in \cite{Lijphart}.

Kemeny-optimal rankings minimize the sum of Kendall-Tau distances between the
ranks proposed by individual voters, and the aggregate rank. 
The problem of computing Kemeny-optimal rankings is known to be NP-hard
\cite{bartholdi,dwork2001rank}.
Cynthia Dwork, Ravi Kumar et al.\ \cite{dwork2001rank} study approximation
methods that can be applied to the problem of ranking search results by
combining the output of several rankers.
Nir Ailon et al. \cite{ailon_2005} developed an algorithm to
find approximate solution subject to additional constraints.
The problem of finding Kemeny optimal solution is equivalent to the 
minimum feedback arc set problem, and Kenyon-Mathieu and Schudy
\cite{kenyon-mathieu_2007} obtained polynomial time algorithm for computing a
solution with loss at most $(1 + \epsilon)$. 

On-line algorithms for rank aggregation have been long studied, especially
in their application to ranking in sports such as chess and tennis.
In these algorithms, a global ranking is gradually refined and updated
according to a stream of incoming comparisons.
In sports, these comparisons consist in the outcomes of matches between
players; in other settings, the comparisons may be obtained by asking users or
visitors to sites to select a winner among a set of alternatives. 
In the original paper by Bradley and Terry, the player strenghts are obtained
from match outcomes via a maximum-likelyhood approach
\cite{bradley_terry,luce}; Elo replaced this with a dynamic update process
which could account also for the time-varying aspect of player strenghts
\cite{ELO}.
Glickman then refined the models and the algorithms by first adapting a
Bayesian update approach \cite{glickman1993paired}, and by then obtaining
efficient algorithms via approximation and parameter estimation
\cite{glickman1999}.

%% file: crowdgrader.tex
\section{CrowdGrader}

CrowdGrader lets students submit and collaboratively grade
solutions to homework assignments.
The lifecycle of an assignment in CrowdGrader consists of three phases: a
submission phase, a review phase, and a grading phase. 

The submission phase is standard.

In the review phase, each student must review a given number of submissions. 
The more submissions each student reviews, the more accurate the
crowd-sourced grade will be, but the larger the workload on the students. 
In our experiments, asking that each submission was reviewed by 5 or
more students yielded acceptable accuracy.

Once the review period is over, CrowdGrader computes a consensus grade for each
submission, aggregating the grades or comparisons provided by the students via
the algorithms we will present in Section~\ref{sec-algos}.
Crowdgrader then assigns a {\em ``crowd-grade''} to each student, by
combining the consensus grade of the submission with a {\em review
grade\/} which quantifies the review effort and accuracy of each student.
In our experiments, computing the crowd-grade by giving  75\% weight to the
submission grade, and 25\% to the review grade, provided sufficient motivation
for the students to put adequate effort in reviewing. 
The instructors can either use the crowd-grade as the grade for the
student in the assignment, or they can fine-tune the final grades, for instance
to correct overall biases.

We applied CrowdGrader to the grading of coding assignments, namely,
Android programming assignments (CMPS~121, taught by one of the authors at
UCSC); C++ programming assignments (CMPS~109, also taught at
UCSC); and Java assignments (taught at University of Naples). 
While CrowdGrader can support in principle many types of assignments, we
focused on programming assignments for three reasons. 

Programming assignments are especially burdensome to grade: unpacking,
compiling, and testing each submission is a time-consuming process.
CrowdGrader enabled us to give coding assignments weekly, spreading what would
have been a very onerous grading task on the students participating in the
class. 

Second, we thought that students would be able to test and evaluate the
submitted code with reasonable accuracy.

Third, we believed that students would directly benefit from reading the code
submitted by other students. 
Strong students would be presented with alternative ways of solving the
problems, and weaker students would have an opportunity to study
several working solutions.
Indeed, students reported a positive experience from the tool,
citing their ability to learn from others, and at the usefulness of the
feedback they received, as the main benefits.

The code for CrowdGrader as used for this paper is available from
\url{https://github.com/lucadealfaro/crowdranker}, and CrowdGrader itself is
available at \url{http://www.crowdgrader.org/}.

%% file: review.tex
\section{Design of the Review Phase}
\label{sec-review}

The review phase is of primary importance for the accuracy of the generated
ranking, and we experimented with several designs. 

\subsection{Review assignment}

We opted for an anonymous review process, in which submissions to review were
assigned automatically to students. 
Since students could not choose which submissions to review, nor in general did
they know the identity of the submissions' authors, they had limited ability to
collude and cause their friends to receive higher grades.

In usual computer-science conferences, papers are assigned to program-committee
members in a single batch; each member then has a period of time to read the
papers and enter all reviews. 
We decided to follow a different approach, in which submissions were assigned
to students for review one at a time: students were assigned a new review task
only upon completion of the previous one.
Our chief concern in making this decision was to ensure that students would not
get the submissions, and their reviews, mixed up. 
Unlike conference papers, the submitted homework solutions are all
on the same topic, and they can be fairly similar to each other;
furthermore, to preserve anonymity, submissions under review were denoted by
un-memorable names such as ``Homework~2 Assignment~3''.
By having students work on one review at a time, we hoped to cut down on the
possibility of mix-ups. 
Indeed, we received no valid reports of mis-directed
reviews.

Delaying the review assignment until the last moment offered two
additional benefits.
First, we were able to ensure that all submissions received roughly the same
number of reviews, even if some students failed to do any reviewing work. 
For each submission, we considered the number of {\em likely reviews,}
consisting of the completed reviews, along with the review tasks that had been
assigned only a short time before. 
When assigning reviews, we chose submissions having least number of likely reviews.
Second, the delayed assignment let us gather information about the
quality of submissions, as the review process proceeded, enabling us to
optimize the review assignment by routing submissions to students who were in
the best position to provide feedback on them.
We have experimented with various techniques for routing submissions, but we do
not yet have sufficient experimental evidence to report on the performance of
the algorithms.

\subsection{Comparisons vs. grades}

For the first homework assignment conducted using CrowdGrader, we decided
to ask students to rank homework submissions, rather than grade them. 
We had more faith in the students' ability to compare submissions,
than in their ability to assign grades with sufficient consistency, so that grades
assigned by different students would be comparable. 
When reviewing a submission, students were presented with a screen displaying
the submissions they had already ranked, in the quality order they had
previously entered; at the bottom, and in a highlighted color, was the
new submission to review. 
Students were instructed to write some feedback for the submission's author,
and then to drag and drop the new submission into the appropriate place in the ranking.\footnote{While inserting the new
submission in the ranking, the students were able to re-order previously ranked
submissions.}

Unfortunately, after writing the feedback paragraph, many students skipped the
ranking step, leaving the new submission where they found it --- at the bottom of the
ranking. 
To confirm this, we measured the fraction of times $f_h$ that students would
rank the newly assigned submission higher than a given submission they had already reviewed. 
Had students been accurate, this fraction should have been close to 50\%, since
there was no relationship between the quality of the new submissions, and that
of the previously-reviewed ones.
Instead, in the first assignment this fraction was only 36\%. 
Even after strongly reminding students to provide a ranking, the fraction $f_h$
rose only to $41\%$ in the second assignment. 
Table~\ref{table-fh} reports the value of $f_h$ for the five CMPS 121 Android
assignments.

\begin{table}[t]
\begin{center}
\begin{tabular}{r|r|r|}
Assignment & $f_h$ & Number of pairs \\ \hline
CMPS 121 hw 1 & 36\% & 252 \\
CMPS 121 hw 2 & 41\% & 231 \\
CMPS 121 hw 3 & 53\% & 271 \\
CMPS 121 hw 4 & 52\% & 277 \\
CMPS 121 hw 5 & 49\% & 221 \\
\end{tabular}
\end{center}
\caption{Fraction $f_h$ of pairs consisting of a previously-reviewed
submission, and a submission under review, in which the submission under review was ranked higher by the student than the previously-reviewed one.}
\label{table-fh}
\end{table}

Talking to students, we understood that they were skipping the ranking step
because of a combination of forgetfulness, and unwillingness. 
Several students mentioned that they felt unconfortable with providing a
ranking of their peers.
Furthermore, they thought that ranking was a blunt instrument.
They complained about having to arbitrarily rank submissions that
they felt were roughly equivalent, and they worried that ranking did not
differentiate between the situations of submissions of roughly equivalent
quality, and submissions of widely different quality.
While ranking can indeed be precise, we are concerned not only with precision,
but also with how the tool is received by the students.

The problem in our UI, of course, was that we could not distinguish
between a skipped ranking, and a valid ranking. 
Starting from the third homework assignment, we modified the UI so that
students needed to {\em both\/} rank the submissions, {\em and\/} assign a
grade to each one: the ranking had to reflect the grades. 
As students could not leave grades blank, this effectively forced students to
provide a valid ranking. 
Table~\ref{table-fh} shows that from assignment~3 onwards the fraction
$f_h$ was very close to 50\%.
Adding grades led to a more accurate ranking of the submissions ---
regardless of whether the grades themselves were used! 
The student satisfaction with CrowdGrader also markedly increased, once grades
were seen as the primary method of providing input to the tool.

Once grades were available, we decided to use the
additional information they convey, and we focused on the development of
crowdsourcing algorithms for the aggregation of grades.
In the current UI of CrowdGrader, students still need to both rank
submissions, and assign them a grade. 
Obviously, once we have grades, the ranking step is un-necessary. 
However, we believe that asking students to also rank the submissions forces
them to consider the relationship between submissions with similar grades,
leading them to fine-tune the grades to more accurately reflect their quality
assessment.
We intend to confirm this belief in future work, comparing the accuracy of the
grades with, and without, the ranking step.

\subsection{Rejecting evaluations}

We discovered early on that it was important to allow students to leave some
submissions ungraded, and yet consider their reviewing duty for the submission
as completed, as far as the computation of the students' own review grades were concerned.
In our programming assignments, there were many cases in which 
well-intentioned students were unable to review submissions. 
In the Android class, their installation of Eclipse and Android SDK
occasionally misbehaved in a way that left students unable to load and review
the code submitted by other students.
In the C++ class, glitches or differences in the build environment
occasionally prevented students from compiling and executing the submissions
under review. 
Initially, students needed to enter a grade to receive credit for
their review effort, and students entered very low grades for
the submissions they could not evaluate. 
In our informal analysis of the accuracy of the tool, this was the
largest source of discrepancy in the grades assigned by
different students to the same submission. 
The solution was to let students flag a review task as ``declined", omitting
the grade, and providing instead an explanation of why they were declining it.
In our experiments, no more than 1\% of submissions required instructor
evaluation, since all students declined their review; these submissions
typically were markedly incomplete and non-functional.

%% file: crowdsourcing.tex
\section{The Vancouver Crowdsourcing Algorithm}
\label{sec-algos}

Once students assign grades to the submissions they review, we need to
aggregate the student-provided grades into a {\em consensus\/} grade for each
submission.
The simplest algorithm for computing consensus grades consists in
averaging the grades each submission has received; we refer to this algorithm
as \algoavg.
We developed an alternative algorithm, the \vanc\ algorithm.\footnote{The
algorithm owes its name to the fact that it was conceived while strolling the
pleasant streets of this Canadian city.} 
The \vanc\ algorithm measures each student's grading accuracy, by comparing the
grades assigned by the student with the grades compared to the same submissions
by other students, and gives more weight to the input of students with higher
measured accuracy. 
The algorithm thus implements a reputation system for students, where higher
accuracy leads to higher reputation, and to higher influence on the consensus
grades.

On synthetic data, \vanc\ is far more accurate than \algoavg.
On our experimental data, \vanc\ performs better than \algoavg, but as we will
report in Section~\ref{sec-results}, the difference is not quite as large,
perhaps due to the fact that our assumptions about user behavior do not fully
correspond to how students behave in practice.
%

\subsection{Variance minimization principle}

The \vanc\ algorithm is based on the following fact.

\begin{prop}{(minimum variance estimator)}
Suppose we have available uncorrelated estimates $\hat{X}_1, \ldots, \hat{X}_n$
of a quantify $x$ of interest, where each $\hat{X}_i$ is a random variable
with average $x$ and variance $v_i$, for $1 \leq i \leq n$.
We can obtain an estimate of $x$ that has minimum variance by averaging
$\hat{X}_1, \ldots, \hat{X}_n$ while giving each $\hat{X}_i$ a weight
proportional to $1 / v_i$, for $1 \leq i \leq n$.
That is, the minimum variance estimator $\hat{X}$ of $x$ can be obtained as:
\[
    \hat{X} = \frac{\sum_{i=1}^n \hat{X}_i /v_i}{\sum_{i=1}^n 1/v_i} \eqpun . 
\]
The variance of this estimator is 
\[
    \V(\hat{X}) = \left( \sum_{i=1}^n \frac{1}{v_i} \right)^{-1} \eqpun .
\]
\end{prop}

\begin{proof}
Given two uncorrelated estimates $\hat{X}_1, \hat{X}_2$, with variances $v_1,
v_2$, consider their linear combination $Y = \alpha_1 \hat{X}_1 +
\alpha_2 \hat{X}_2$, with $\alpha_1 + \alpha_2 = 1$.
By the Bienaym\'e formula, the variance of $Y$ is given by 
$\alpha_1^2 v_1 + (1 - \alpha_1)^2 v_2.$.
If we take the derivative with respect to $\alpha_1$, and set it to~0, we
obtain $\alpha_1 v_1 = \alpha_2 v_2$, or $\alpha_1 \propto 1/v_1$ and
$\alpha_2 \propto 1/v_2$.
The general case for $n$ estimates follows similarly. 
\end{proof}

This observation immediately suggests how to obtain
reputation-based crowdsourcing algorithms for grades: if we could somehow
measure the variance $v_i$ of each student $i$, we could weigh the input
provided by student $i$ in proportion to $1 / v_i$.

\subsection{Algorithm structure}

We developed an algorithm that proceeds
in iterative fashion, using consensus grades to estimate the grading
variance of each user, and using the information on user variance to compute
more precise consensus grades. 
The structure of the algorithm is inspired by the algorithm of
\cite{Karger2011} for computing consensus boolean values.
To state the algorithm, we denote by $U$ the set of students, and by $S$ the
set of items to be graded (the submissions). 
We let $G = (T, E)$  be the graph encoding the review
relation, where $T = S \union U$ and $S \inters U = \emptyset$, and where 
$(i, j) \in E$ iff $j$ reviewed $i$; for $(i, j) \in E$, we let $g_{ij}$ be the
grade assigned by $j$ to $i$.
We denote by $\partial t$ the 1-neighborhood of a node $t \in T$.

The algorithm proceeds by updating estimtes $v_j$ of the variance of user $j
\in U$, and estimates $c_i$ of the consensus grade of item $i \in S$, and
estimates $v_i$ of the variance with which $c_i$ is known.
To produce these estimates, the algorithm relies on messages $m = (l, x, v)$
consisting of a source $l \in S \union U$, of a value $x$, and of a variance
$v$.
We denote by $M_i, M_j$ the lists of messages associated with item $i \in
T$ or user $j \in U$.
Given a set $M$ of messages, we indicate by 
\begin{align*}
    \E(M) & = \frac{\sum_{(l, x, v) \in M} x / v}{\sum_{(l, x, v) \in M} 1/v}
    \\
    \V(M) & = \left(\sum_{(l, x, v) \in M} \frac{1}{v}\right)^{-1}
\end{align*}
the best estimator we can obtain from $M$, and its variance.

\begin{algorithm*}[t]
    \caption{The Vancouver Algorithm.}
    \label{algo-vanc}
    {\bf Input:} A review graph $G = ((S \union U), E)$ such that
    $|\partial t| > 1$ for all $t \in S \union U$, along with
    $\set{g_{ij}}_{(i,j) \in E}$, and number of iterations $K > 0$. \\
    {\bf Output:} Estimates $\hat{q}_i$ for $i \in S$.
    \begin{algorithmic}[1]
        \STATE \COMMENT{Initialization}
        \FORALL{$i \in S$} \label{algo-vanc-ini1}
            \STATE{$M_i := \set{(j, g_{ij}, 1) \mid (i, j) \in E}$.}
        \ENDFOR \label{algo-vanc-ini3}
        \FOR{iteration $k = 1, 2, \ldots, K$}
            \STATE \COMMENT{Propagation from items}
            \FORALL{$j \in U$} \label{algo-vanc-from-items-1}
                \STATE{$M_j := \emptyset$}
            \ENDFOR
            \FORALL{$i \in S$}
                \FORALL{$j \in \partial i$} \label{algo-vanc-opt1}
                    \STATE{Let $M_{-j} = \set{(j', x, v) \in M_i \mid
                            j' \neq j}$ in
                            $M_j := M_j \union (i, \E(M_{-j}), \V(M_{-j}))$}   
                            \label{algo-vanc-set1}
                \ENDFOR
            \ENDFOR \label{algo-vanc-from-items-10}
            \STATE \COMMENT{Propagation from users}
            \FORALL{$i \in S$} \label{algo-vanc-from-users-1}
                \STATE{$M_i := \emptyset$}
            \ENDFOR
            \FORALL{$j \in U$}
                \FORALL{$i \in \partial j$} \label{algo-vanc-opt2}
                    \STATE{Let $M_{-i} = \set{(i', (x - g_{i'j})^2, v) \mid 
                            (i', x, v) \in M_j, i' \neq i}$ in
                            $M_i := M_i \union (j, g_{ij}, \E(M_{-j}))$}
                            \label{algo-vanc-set2}
                \ENDFOR
            \ENDFOR \label{algo-vanc-from-users-10}
        \ENDFOR
        \STATE \COMMENT{Final Aggregation}
        \FORALL{$i \in S$} \label{algo-vanc-aggreg-1}
            \STATE{$\hat{q}_i := \E(M_i)$}
        \ENDFOR \label{algo-vanc-aggreg-3}   
    \end{algorithmic}
\end{algorithm*}

The details are given in Algorithm~\ref{algo-vanc}.
Lines \ref{algo-vanc-ini1}--\ref{algo-vanc-ini3} initialize the messages to
items using the grade assigned by the users, and a constant variance (whose
precise value is unimportant).  
If we had a-priori information on the variance of some users, it could be used
in this initialization step.
Lines \ref{algo-vanc-from-items-1}--\ref{algo-vanc-from-items-10} propagate,
from items to the users who graded them, the best estimate
available on the item grades and variances.  
In line~\ref{algo-vanc-set1}, when we compute the estimate that is sent to each
user, we do not use information coming from that same user.
Lines \ref{algo-vanc-from-users-1}--\ref{algo-vanc-from-users-10} propagate,
from users to the items they graded, the (immutable) grade the user assigned to
the item, and a newly-recomputed estimate of the user's grading variance.
The estimate of the user variance is computed by considering the differences
between the item grades assigned by the user, and the estimates received from
the items.
Again, when computing the user variance that will be sent to an item, we do not
consider the contribution to the variance due to this same item.
Finally, in lines \ref{algo-vanc-aggreg-1}--\ref{algo-vanc-aggreg-3} we
aggregate the information from users into our final estimates of item grades.
We note that we gave above the most concise presentation of the algorithm; a
more efficient implementation can be obtained by optimizing, in the loops at
lines \ref{algo-vanc-opt1} and \ref{algo-vanc-opt2}, the constructions of the
sets of messages, considering the overlap between the sets. 
This reduces the time for each loop from $\calo(n m^2)$ to $\calo(n m)$, where
$n$ is the number of users and items, and $m$ is the number of reviews for each
item.

\subsection{Performance on synthetic data}

We evaluated the performance of \vanc\ on simulated data; results
on real-world data will be given in Section~\ref{sec-results}.
We considered 50 users and 50 items, with each user reviewing 6 items; these
numbers are similar to those occurring in our actual assignments.
The true quality $q_i$ of each item $i$ we assumed was normal-distributed with
standard deviation~1. 
We assumed that each user $j$ had a characteristic variance $v_j$, and we let
the grade $q_{ij}$ assigned by $j$ to $i$ be equal to $q_i + \Delta_{ij}$, 
where $q_i$ is the true quality of $i$, and $\Delta_{ij}$ has normal
distribution with mean 0 and variance $v_j$. 
We assumed that the variances $\set{v_j}_{j \in U}$ of the users were
distributed according to a Gamma distribution with scale~0.4, and shape
factors $k = 2, 3$.
The results are summarized in Table~\ref{table-vancouver}. 
For each shape factor, and each of the two algorithms \algoavg and \vanc, we
report the statistical correlation $\rho$ between true quality $q_i$ and
consensus quality $\hat{q}_i$ for all items $i$, as well as the standard
deviation $\sigma$ of the difference $q_i - \hat{q}_i$. 
Each entry in the table is the average over 100 runs.
The \vanc\ algorithm reduces the error betewen true and consensus grades by a
factor between 3 and~4, compared with simple average \algoavg.
The fact that the gain is larger for shape factor $k=2$ compared with $k=3$
indicates that the algorithm performs better when there are fewer, more
imprecise users. 
Even more significant is the increase in the correlation $\rho$.
The code used for the table can be obtained from
\url{https://github.com/lucadealfaro/vancouver}, and corresponds to the tag
``2013-techrep''; the code can be easily adapted to study the performance of
the algorithms under different sets of assumptions on user behavior.

\begin{table}
\begin{center}
\begin{tabular}{r||r|r||r|r|}
 & \multicolumn{2}{|c||}{$\rho$} & \multicolumn{2}{|c|}{$\sigma$} \\ 
 & \multicolumn{1}{|c|}{$k=2$} & \multicolumn{1}{|c||}{$k=3$}
 & \multicolumn{1}{|c|}{$k=2$} & \multicolumn{1}{|c|}{$k=3$} \\ \hline
\algoavg & 0.82 & 0.63 & 0.69 & 1.21 \\
\vanc    & 0.99 & 0.93 & 0.15 & 0.38 \\ \hline
\end{tabular}
\end{center}
\caption{Performance of \vanc\ algorithm on synthetic data.}
\label{table-vancouver}
\end{table}

%% file: results.tex
\section{Experimental Results}
\label{sec-results}

We performed two different types of evaluations of the precision of the \vanc\
algorithm in assigning consensus grades to assignments.
In one type of evaluation, we compared crowdsourced consensus grades with
control grades given by the instructor or other domain experts; 
in the other type, we measured the grade difference among submissions
that we knew were identical.

\subsection{The dataset}

The evaluation dataset consisted in five homework assignments for an Android
class (CMPS 121); five homework assignments for a C++ class (CMPS 109), and one
homework assignment for a Java class (LP2).
The number of homework submissions, and reviews, for these classes are
summarized in Table~\ref{table-num-reviews}.
As the table indicates, students generally performed the reviews that they were
asked to do, indicating that the system of incentives we have in place
(discussed more in depth in Section~\ref{sec-incentives}) was effective. 
Some of the difference between the number of reviews due, and performed, can be
ascribed to the fact that students could decline to review specific
submissions.
The table also shows that, in the initial homework assignments of each class,
some submissions received a low number of reviews. 
This occurred as we had not yet fine-tuned our algorithms for assigning reviews
to students. 
Once we developed algorithms that try to predict the probability that each
outstanding review will be completed, we were able to ensure a more uniform
review coverage.

\begin{table}[t]
\begin{center}
\begin{tabular}{r|r|r|r|r|r|}
Assignment & $|S|$ & RevsDue & MinRevs & AvgRevs \\ \hline
CMPS 121 hw 1 & 60 & 6 & 2 & 5.4 \\
         hw 2 & 61 & 6 & 2 & 5.3 \\
         hw 3 & 68 & 6 & 0 & 4.8 \\
         hw 4 & 62 & 6 & 6 & 6.1 \\
         hw 5 & 57 & 6 & 5 & 5.3 \\ \hline
CMPS 109 hw 1 & 102 & 5 & 0 & 4.6 \\ 
         hw 2 &  97 & 5 & 3 & 4.6 \\
         hw 3 &  91 & 5 & 4 & 5.1 \\
         hw 4 &  97 & 5 & 3 & 4.6 \\
         hw 5 &  90 & 5 & 4 & 5.1 \\ \hline
\end{tabular}
\caption{Number of reviews assigned and performed for the homework assignments
that are part of the dataset.  $|S|$ is the number of submissions, RevsDue is
the number of reviews that each student ought to have done, MinRevs is the
minimum number of reviews received by a submission, and AvgRevs is the
average number of reviews per submission.}
\label{table-num-reviews}
\end{center}
\end{table}

\subsection{Evaluation using control grades}

For some assignments, we had available control grades given by the instructor,
or other domain experts, for a randomly selected subset of submissions that
numbered at least~20.
For the Android assignments, the control grades were assigned by a Teaching
Assistant (TA) who was a fairly accomplished Android developer. 
For the Java assignment, the control grades were provided by the instructor. 
For the C++ assignments, the authors graded 20 or more randomly selected
submissions for each assignment.
We compared the control grades with the consensus grades computed by \algoavg\
and \vanc\ according to the following metrics:
\begin{itemize}
  \item $\rho$: the coefficient of statistical correlation (also
  known as Pearson's correlation) between the control grades $\set{q_i}$ and
  the consensus grades $\set{\hat{q}_i}$.
  \item KT: the Kendall-Tau distance between the orderings induced
  by the control and consensus grades \cite{kendall1990rank}.
  If $r_i$ and $t_i$ are the ranks received by submission $i$ in the computed,
  and control, rankings respectively, then $\mbox{KT} = \sum_i (r_i - t_i)$.
  \item norm-2: the norm-2 distance $(\sum_i (q_i - \hat{q}_i)^2)^{1/2}$
  between the control grades $\set{q_i}$ and the consensus grades
  $\set{\hat{q}_i}$.  Grades were awarded on a scale from 0 to~10 in the
  assignments.\footnote{The grading scale can be chosen for each assignment,
  but all assignments so far have used a 0 to 10 scale.}
  \item s-score: we first normalize the control grades $\set{q_i}$ and the
  consensus grades $\set{\hat{q}_i}$, so that they both have zero mean and unit
  variance, obtaining $\set{q'_i}$, $\set{\hat{q}'_i}$   Then, we compute the
  standard deviation $s$ of $\set{q'_i - \hat{q}'_i}$, and we report the
  s-score $1 - s / \sqrt{2}$.
\end{itemize}

The results for the various assignments are reported in
Table~\ref{table-control}.
We see that the results are unclear: in CMPS~121 and JP2, \vanc\ does better;
in the two CMPS~109 assignments, it does worse.
This may be a consequence of the fact that the primary cause of evaluation
error in CMPS~109 consisted in failures encountered by students
in compiling the C++ submissions of other students, triggered by development
environment (operating system, build chain) differences. 
These failures are not well modeled by the assumption that each user has
an intrinsic review accuracy: the fact that compilation problems occurred in
one review may have little bearing on the accuracy of other reviews by the same user.
The low correlation between consensus grades and control grades for CMPS~121 is
due to the fact that the control grades have a very coarse granularity (few
values in the grading scale were used).
We also note that this evaluation is inherently approximate, since the control
grade is affected by the same type of imprecision that affects the student-provided
grades.
While instructor and TAs are (usually) more knowledgeable than students in the
subject matter, they also make mistakes when grading homeworks, failing to spot
problems, or not giving credit to great aspects of the work that go undetected.  

\begin{table*}[t]
\begin{center}
\begin{tabular}{ll|r|r|r|r|}
Homework & Algorithm & $\rho$ & KT & norm-2 & s-score \\ \hline 
CMPS 109 hw 2 & \algoavg\ & 0.75 & 0.37 & 1.40 & 0.50 \\ 
              & \vanc\    & 0.69 & 0.39 & 1.59 & 0.45 \\ \hline
CMPS 109 hw 3 & \algoavg\ & 0.84  & 0.39 & 1.49 & 0.60 \\
              & \vanc\    & 0.80 & 0.42 & 1.75 & 0.55 \\ \hline
CMPS 121 hw 3 & \algoavg\ & 0.39 & 0.53 & 1.63 & 0.22 \\
              & \vanc\    & 0.49 & 0.53 & 1.33 & 0.29 \\ \hline
LP2           & \algoavg\ & 0.85 & 0.20 & 1.75 & 0.61 \\
              & \vanc\    & 0.87 & 0.18 & 1.79 & 0.64 \\ \hline
\end{tabular}
\end{center}
\caption{Performance of \algoavg\ and \vanc, with respect to control grades.}
\label{table-control}
\end{table*}

\subsection{Evaluation using pairs of identical submissions}

For some of the CMPS~109 C++ homework assignments, students were able to work
in groups.
Since at the time CrowdGrader did not support group submissions (the feature
has since been added), the students were asked to each submit a solution. 
The student submissions would be graded independently, and the TA, who had a
list of groups and their members, would then average the grades received by the
students in the same group, and assign to each group member this average.
This meant that we had available several pairs of identical submissions, coming
from members of the same group.
This made it possible to judge the quality of a crowdsourcing algorithm
according to how close were the grades received by pairs of such identical
submissions.
In Table~\ref{table-pairs}, we report on the average $D$ of $(\hat{q}_i -
\hat{q}_l)^2$, computed over all pairs $(i, l)$ of identical submissions, for
the algorithms \vanc\ and \algoavg.
We see that according to this measure, even for CMPS~109 \vanc\ has generally
better performance than \algoavg, even though the difference is not large.

\begin{table}[t]
\begin{center}
\begin{tabular}{l|r|r|r|}
Assignment & D, \vanc\ & D, \algoavg\ & N. pairs \\ \hline
CMPS 109 hw 2 & 1.97 & 3.24 & 6 \\ 
CMPS 109 hw 3 & 1.29 & 1.39 & 12 \\ 
CMPS 109 hw 4 & 0.98 & 1.07 & 20 \\ 
CMPS 109 hw 5 & 1.38 & 1.19 & 20 \\ \hline
\end{tabular}
\caption{Average square difference between grades received by identical
assignments, using crowdsourcing algorithms \vanc\ and \algoavg.}
\label{table-pairs}
\end{center}
\end{table}

\subsection{Discussion}

The results presented in this section show that, for our assignments, the
\vanc\ algorithm provides a smaller advantage, compared to \algoavg, than it
would be expected from Table~\ref{table-vancouver}.
We believe that the lower performance is due to the fact that the user error
model used in developing algorithm \vanc, in which each user $i$ has a
variance $v_i$, is only an approximation for the real behavior of students
reviewing submissions. 
The largest single cause of review errors were:
\begin{itemize}
  \item Unclear problem statements, that caused different students to have
  different interpretations of what constituted a good homework solution.
  \item Variability in the student's code development environment that
  occasionally prevented students from compiling and evaluating
  submissions.
\end{itemize}
The clarity and precision of homework assignments is likely the major factor in
the precision of any tool, or any TA, in evaluating submitted solutions.
We believe that the higher correlation and quality of the results for the Java
assignment are due to the uniformity of the environment enforced for that
submission.

We also experimented with a number of variations of algorithm \vanc, some based
on using notions of median or weighed median for selecting grades.
In particular, we experimented with a method we nicknamed ``maverage'', in
which we aggregated student-assigned grades for each item by first discarding
the highest and lowest grades, then doing a weighted average using the
reciprocal of variance as weights. 
This process was inspired by the way used to average the grades given by
Olympic judges in competitions. 
We also tried to learn the positive or negative bias of each student compared
to the others, and subtract the bias before using the student's grades. 
None of these variants was clearly superior to \vanc.
We believe that larger datasets are needed for us to be able to formulate and
validate algorithms superior to \vanc.

%% file: incentives.tex
\section{Review Incentive and Final Grade Assignment}
\label{sec-incentives}

\subsection{Review Incentive}

To provide an incentive for students to complete a certain number of reviews
per assignment, we made the review effort a component of the overall grade that
was assigned to students. 
For each homework assignment, the instructor could choose the number $N$ of
reviews each student had to perform, and the fraction $0 < p_r < 1$ of the
grade that was due to reviews. 
Each student $j$ then received for the assignment a {\em crowd-grade\/} equal
to
\[
    (1 - p_r) \hat{q}_j + p_r \frac{\min(m_j, N)}{N} \hat{r}_j \eqpun ,
\]
where $m_j$ is the number of reviews actually performed by student
$j$, and where $\hat{r}_j$ is the estimated {\em review quality\/} of $j$,
which we discuss below.
The choice of $N$ and $p_r$ was dictated chiefly by practical considerations. 
In our coding assignments, evaluating a homework submission entailed a lengthy
process of unpacking a submission in its own directory, loading it with a tool,
reading the various source code files, compiling it, and testing it sometimes
with the help of test data. 
The whole process would take between 5 and 10 minutes for each homework;
we chose $N = 5$ or $N = 6$, as the results appeared to be sufficiently
accurate.  
We also wanted to ensure that each student was able to learn by
reading good-quality submissions by others, and a value of $N = 5$ was
sufficient in practice to ensure this (students could always do additional
reviews if they wished to see even more solutions).
For $p_r$, a common choice was 0.25, so that 25\% of the crowd-grade was due to
the reviews.
This value roughly reflected the proportion between the time required to
review the submissions, and the time required to complete and submit one's own submission.

The decision of how to measure $\hat{r}_j$ for a student $j$ turned out to be
more difficult. 
Initially, we defined it as follows. 
Let $\set{g_{ij}}_{i \in S, j \in U}$ be the set of all grades that were
assigned, and let $\set{\hat{q}_i}_{i \in U}$ be the set of all consensus
grades, as before.
Then, 
\[
    \tilde{v} = \E(\set{(g_{ij} - \hat{q}_l)^2}_{j \in U; i, l \in S})
\]
is the average square error with respect to the consensus grades of a
hypothetical ``fully random'' user, who assigns to each submission a grade
picked at random from the complete set of assigned grades.
The actual average square error  $\tilde{v}_j$ of a student $j \in U$ with
respect to the consensus grades is instead: 
\[
    \tilde{v}_j = \E(\set{(g_{ij} - \hat{q}_i)}_{i \in S}) \eqpun . 
\]
Therefore, we experimented with assigning to each student a review grade
that measured how much better the student was than such a fully random grader,
using:
\[
    \hat{r}_j = 1 - \sqrt{\frac{\min(\tilde{v}_j, \tilde{v})}{\tilde{v}}}
    \eqpun .
\]
This choice appealed to us from a theoretical point of view, especially as it
is scale-invariant, so that it would not matter whether students were using the
full grading scale (in our case, $[0, 10]$) or a subset of it (for instance,
assigning grades only in the interval $[4, 8]$).
However, the choice did not work to our satisfaction in practice. 
In each assignment, some perfectly honest and motivated students
received very low review grades, including~0: strange
as it might seem, some students really did worse than a random grader, in spite
of their best intentions. 
Those students were not pleased to see the time they put into reviewing
homework submissions go completely unrewarded.
The problem was especially acute in the initial homework assignments of each
class, where a large fraction of homework submissions received the maximum
grade, thereby lowering $\tilde{v}$, and making it harder to improve on the
random grader.

As student satisfaction is one of our goals, we needed a different approach. 
We do not yet have a perfect solution: a metric that is scale invariant and
rewards true accuracy as compared to random input, and  yet, that students find
fair and gratifying.
The metric currently used by CrowdGrader is a fairly generous one.
We let $v_G = G^2 / 3.125$ be a refence level for the average square error,
where $G$ is the maximum of the grading scale used (we omit the
justification, as it is fairly ad-hoc), and we use:
\[
    \hat{r}_j = 1 - \sqrt{\frac{\min{\tilde{v}_j, v_G}}{v_G}} \eqpun .
\]
This is not scale-invariant, so that students would get
a higher review grade simply by agreeing to use only a small portion of the
overall grade range available to them.
We are still seeking a scale-invariant solution that students find equitable.

\subsection{Final grade assignment}

CrowdGrader produces {\em crowd-grades\/} that depend both on the submission and
on the review grades, as described above. 
The instructor can then either accept these grades as final, or provide final
grades for a few of the students; the final grades for the remainder of the students are
then derived by interpolation, according to their crowd-grades. 
This gives the ability to the instructor to re-shape the grade curve of the
class.
In the Android class (CMPS~121), the instructor relied on this function to
manually choose the dividing lines between A/B, B/C, and C/F grades.
The instructor examined several assignments chosen from the class rank
order, read the reviews, and assign grades (5.3 for A+, 4.5 for the A/B
dividing line, etc.) to selected assignments; CrowdGrader then computed the
remaining final grades by linear interpolation, in proportion to the
crowd-grades. 
In CMPS~109, the instructors often used the crowd-grades as final grades.

%% file: conclusions.tex
\section{Conclusions}

We conclude with some informal impressions on the performance of CrowdGrader in
a class setting.

We investigated many cases where the control and consensus
grades differed by some non-trivial amount. 
In some cases, this was due to superficial reviews by students using
CrowdGrader. 
However, in other cases the problem was with the control grade, as the
instructor or TA had missed problems with the submission that were instead
detected by some students reviewing it. 
Overall, for coding assignments, our impression was that the consensus grades
computed by CrowdGrader were at least of the same quality as those provided by
a TA. 
A TA is more consistent in evaluating submissions, paying attention to the same
aspects of each submission. 
On the other hand, the greater number of reviews used in CrowdGrader led to a
more comprehensive assessment, in which flaws or positive aspects were more
likely to be pointed out.
From the perspective of the individual student, we felt the two grading options
were of similar quality: with TAs, the risk is that they do not pay attention
in their grading to the aspects where most effort is put (or where the flaws
are); with crowdsourced grades, the risk is in the inherent variability of
the process.

Where the crowdsourced evaluations proved clearly superior was in the
feedback provided to the students.
When instructors or TAs are faced with grading a large number of assignments,
the feedback they provide on each individual assignment is usually limited.
With CrowdGrader, students had access to multiple reviews of their
homework submissions. 

In coding assignments, there is usually more than one way to solve each
problem, and students commented on the benefit of being able to see, and learn
from, other students' solutions.
Students who could not complete the assignment particularly benefited from
being able to examine several different working solutions to the homework
problems.

In informal comments we received, the two aspects of CrowdGrader students
appreciated the most was the quality of the feedback received, and the ability
to learn from other students' solutions.
The one aspect they enjoyed the least, of course, was the time it took for them
to do the reviews. 

While CrowdGrader may not be suitable for all types of homework assignments,
the tool performed to our satisfaction for coding assignments, and we believe
that the tool is well-suited to any homework assignment where students can, by
comparing solutions among them and with their own, come to an assessment of
their peers' work.

%% file: main.bbl
\begin{thebibliography}{10}

\bibitem{ailon_2005}
N.~Ailon, M.~Charikar, and A.~Newman.
\newblock Aggregating inconsistent information: ranking and clustering.
\newblock In {\em Proceedings of the thirty-seventh annual ACM symposium on
  Theory of computing}, STOC '05, pages 684--693, New York, NY, USA, 2005. ACM.

\bibitem{arrow}
K.~Arrow.
\newblock A difficulty in the concept of social welfare.
\newblock {\em Journal of Political Economy}, 58:328, 1950.

\bibitem{bartholdi}
J.~Bartholdi, C.~Tovey, and M.~Trick.
\newblock Voting schemes for which it can be difficult to tell who won the
  election.
\newblock {\em Social Choice and Welfare}, 6(2):157--165, 1989.

\bibitem{bradley_terry}
R.~Bradley and M.~Terry.
\newblock Rank analysis of incomplete block designs: I. the method of paired
  comparisons.
\newblock {\em Biometrika}, 39(3/4):pp. 324--345, 1952.

\bibitem{das2010ranking}
A.~Das~Sarma, A.~Das~Sarma, S.~Gollapudi, and R.~Panigrahy.
\newblock Ranking mechanisms in twitter-like forums.
\newblock In {\em Proceedings of the third ACM international conference on Web
  search and data mining}, pages 21--30. ACM, 2010.

\bibitem{dawid1979maximum}
A.~Dawid and A.~Skene.
\newblock Maximum likelihood estimation of observer error-rates using the em
  algorithm.
\newblock {\em Applied Statistics}, pages 20--28, 1979.

\bibitem{dealfaro2011reputation}
L.~De~Alfaro, A.~Kulshreshtha, I.~Pye, and B.~Adler.
\newblock Reputation systems for open collaboration.
\newblock {\em Communications of the ACM}, 54(8):81--87, 2011.

\bibitem{Borda}
J.-C. de~Borda.
\newblock {\em Memoire sur les Elections au Scrutin}.
\newblock 1781.

\bibitem{dempster1977maximum}
A.~Dempster, N.~Laird, and D.~Rubin.
\newblock Maximum likelihood from incomplete data via the em algorithm.
\newblock {\em Journal of the Royal Statistical Society. Series B
  (Methodological)}, pages 1--38, 1977.

\bibitem{dwork2001rank}
C.~Dwork, R.~Kumar, M.~Naor, and D.~Sivakumar.
\newblock Rank aggregation methods for the web.
\newblock In {\em Proceedings of the 10th international conference on World
  Wide Web}, pages 613--622. ACM, 2001.

\bibitem{ELO}
A.~Elo.
\newblock {\em The Rating of Chess Players Past and Present}.
\newblock New York, Arco, 1978.

\bibitem{NSFDear2013}
N.~S. Foundation.
\newblock Dear colleague letter: Information to principal investigators (pis)
  planning to submit proposals to the sensors and sensing systems (sss) program
  {October}~1 , 2013 deadline, 2013.

\bibitem{geanakoplos2005three}
J.~Geanakoplos.
\newblock Three brief proofs of arrowÕs impossibility theorem.
\newblock {\em Economic Theory}, 26(1):211--215, 2005.

\bibitem{glickman1993paired}
M.~Glickman.
\newblock {\em Paired Comparison Models with Time-varying Parameters}.
\newblock Harvard University, 1993.

\bibitem{glickman1999}
M.~E. Glickman.
\newblock Parameter estimation in large dynamic paired comparison experiments.
\newblock {\em Journal of the Royal Statistical Society: Series C (Applied
  Statistics)}, 48(3):377--394, 1999.

\bibitem{jin2002learning}
R.~Jin and Z.~Ghahramani.
\newblock Learning with multiple labels.
\newblock In {\em Advances in neural information processing systems}, pages
  897--904, 2002.

\bibitem{Karger2011}
D.~Karger, S.~Oh, and D.~Shah.
\newblock Iterative learning for reliable crowdsourcing systems.
\newblock In {\em Proc.\ of the 25th Annual Conference on Neural Information
  Processing Systems (NIPS)}, 2011.

\bibitem{kendall1990rank}
M.~Kendall and J.~D. Gibbons.
\newblock {\em Rank Correlation Methods}.
\newblock Edward Arnold, 1990.

\bibitem{kenyon-mathieu_2007}
C.~Kenyon-Mathieu and W.~Schudy.
\newblock How to rank with few errors.
\newblock In {\em Proceedings of the thirty-ninth annual ACM symposium on
  Theory of computing}, STOC '07, pages 95--103, New York, NY, USA, 2007. ACM.

\bibitem{Lijphart}
A.~Lijphart.
\newblock {\em Electoral Systems and Party Systems: {A} Study of Twenty-Seven
  Democracies, 1945, 1990}.
\newblock Oxford University Press, 1994.

\bibitem{luce}
R.~Luce.
\newblock {\em Individual choice behavior : a theoretical analysis}.
\newblock Wiley N.Y, 1959.

\bibitem{Merrifield09}
M.~Merrifield and D.~Saari.
\newblock Telescope time without tears: A distributed approach to peer review.
\newblock {\em Astronomy \& Geophysics}, 50(4):4--16, 2009.

\bibitem{Naghizadeh2013}
P.~Naghizadeh and M.~Liu.
\newblock Incentives, quality, and risk: A look into the nsf proposal review
  pilot.
\newblock {\em Arxiv}, 1307.6528v1, 2013.

\bibitem{Pearl1988}
J.~Pearl.
\newblock {\em Probabilistic reasoning in intelligent systems: networks of
  plausible inference}.
\newblock Morgan Kaufmann Publishers Inc., San Francisco, CA, USA, 1988.

\bibitem{raykar_2010}
V.~Raykar, S.~Yu, L.~Zhao, G.~Valadez, C.~Florin, L.~Bogoni, and L.~Moy.
\newblock Learning from crowds.
\newblock {\em J. Mach. Learn. Res.}, 11:1297--1322, Aug. 2010.

\bibitem{smyth_94}
P.~Smyth, U.~Fayyad, M.~Burl, P.~Perona, and P.~Baldi.
\newblock Inferring ground truth from subjective labelling of venus images.
\newblock In G.~Tesauro, D.~S. Touretzky, and T.~K. Leen, editors, {\em NIPS},
  pages 1085--1092. MIT Press, 1994.

\bibitem{welinder_10}
P.~Welinder, S.~Branson, S.~Belongie, and P.~Perona.
\newblock The multidimensional wisdom of crowds.
\newblock In J.~Lafferty, C.~K.~I. Williams, J.~Shawe-Taylor, R.~Zemel, and
  A.~Culotta, editors, {\em Advances in Neural Information Processing Systems
  23}, pages 2424--2432. 2010.

\bibitem{whitehill2009whose}
J.~Whitehill, T.-F. Wu, J.~Bergsma, J.~Movellan, and P.~Ruvolo.
\newblock Whose vote should count more: Optimal integration of labels from
  labelers of unknown expertise.
\newblock In {\em Advances in neural information processing systems}, pages
  2035--2043, 2009.

\bibitem{Yedidia2003}
J.~S. Yedidia, W.~T. Freeman, and Y.~Weiss.
\newblock Exploring artificial intelligence in the new millennium.
\newblock chapter Understanding belief propagation and its generalizations,
  pages 239--269. Morgan Kaufmann Publishers Inc., San Francisco, CA, USA,
  2003.

\end{thebibliography}
